\documentclass[aps,longbibliography,superscriptaddress,nofootinbib]{revtex4-2}
\usepackage{graphicx}
\usepackage{bbold}
\usepackage{hyperref}
\usepackage{amsthm}
\usepackage[leqno]{amsmath}
\usepackage{appendix}
\usepackage[capitalise]{cleveref}
\usepackage{float}
\usepackage{overpic}
\usepackage{tabularx}

\usepackage{xcolor}
\usepackage{amssymb}
\usepackage{mathtools}
\definecolor{riverlane_green}{RGB}{0, 150, 143}
\definecolor{riverlane_orange}{RGB}{255, 117, 0}

\hypersetup{
  colorlinks = true, 
  urlcolor = riverlane_green, 
  linkcolor = riverlane_green, 
  citecolor = riverlane_orange
}

\edef\restoreparindent{\parindent=\the\parindent\relax}
\usepackage{parskip}
\restoreparindent

\newtheorem{theorem}{Theorem}
\newtheorem{proposition}{Proposition}[theorem]
\newtheorem{corollary}[proposition]{Corollary}

\newtheoremstyle{break}
  {\topsep}{\topsep}
  {\itshape}{}
  {\bfseries}{}
  {\newline}{}
\theoremstyle{break}
\newtheorem{algo}{Algorithm}

\begin{document} 

\title{Pauli Decomposition via the Fast Walsh-Hadamard Transform}
\author{Timothy N. Georges}
\affiliation{Riverlane Ltd., St Andrews House, 59 St Andrews Street, Cambridge, CB2 3BZ, United Kingdom}
\affiliation{Department of Chemistry, Physical and Theoretical Chemistry Laboratory, University of Oxford, Oxford, OX1 3QZ, United Kingdom}
\author{Bjorn K. Berntson}
\affiliation{Riverlane Ltd., St Andrews House, 59 St Andrews Street, Cambridge, CB2 3BZ, United Kingdom}
\author{Christoph S\"underhauf}
\affiliation{Riverlane Ltd., St Andrews House, 59 St Andrews Street, Cambridge, CB2 3BZ, United Kingdom}
\author{Aleksei V. Ivanov}
\email[All authors contributed equally. Corresponding author: ]{aleksei.ivanov@riverlane.com}
\affiliation{Riverlane Ltd., St Andrews House, 59 St Andrews Street, Cambridge, CB2 3BZ, United Kingdom}

\begin{abstract}
    The decomposition of a square matrix into a sum of Pauli strings is a classical pre-processing step required to realize many quantum algorithms. Such a decomposition requires significant computational resources for large matrices. We present an exact and explicit formula for the Pauli string coefficients which inspires an efficient algorithm to compute them. More specifically, we show that up to a permutation of the matrix elements, the decomposition coefficients are related to the original matrix by a multiplication of a generalised Hadamard matrix. This allows one to use the Fast Walsh-Hadamard transform and calculate all Pauli decomposition coefficients in $\mathcal{O}(N^2\log N)$ time and using $\mathcal{O}(1)$ additional memory, for an $N\times N$ matrix. A numerical implementation of our equation outperforms currently available solutions.
\end{abstract}

\maketitle

\section{Introduction}\label{sec:introduction}

The Pauli matrices $\mathcal{P}\coloneqq\{I,X,Y,Z\}$, where
\begin{equation}
	I\coloneqq\left(\begin{array}{cc}
        1 & 0 \\
        0 & 1
    \end{array}\right),\quad 
    X\coloneqq\left(\begin{array}{cc}
        0 & 1 \\
        1 & 0
    \end{array}\right),\quad 
    Y\coloneqq\left(\begin{array}{cc}
        0 & -i \\
        i & 0
    \end{array}\right),\quad 
    Z\coloneqq\left(\begin{array}{cc}
        1 & 0 \\
        0 & -1
    \end{array}\right),
\end{equation}
form a set of elementary gates, ubiquitous in the field of quantum computation and information~\cite{nielsen_quantum_2021}.  In order to apply a quantum algorithm to a specific application area such as simulation of quantum systems~\cite{abrams_simulation_1997,abrams_quantum_1999,ortiz_quantum_2001,aspuru_guzik_simulated_2005} or solving systems of linear equations~\cite{harrow_quantum_2009}, it is often desirable to express the problem input, such as a Hamiltonian or more generally a complex matrix, $A\in \mathbb{C}^{2^n\times 2^n}$ with a positive integer $n$, as a weighted sum of tensor products of Pauli matrices, also known as Pauli strings: 
\begin{equation}\label{eq:Aalpha}
A= \sum_{P\in \mathcal{P}_n} \alpha(P) P,
\end{equation}
where 
\begin{equation}
\mathcal{P}_n\coloneqq \Bigg\{ \bigotimes_{j=0}^{n-1} P_j: P_j\in \mathcal{P}	 \Bigg\}
\end{equation}
and 
$\alpha:\mathcal{P}_n\to \mathbb{C}$. The problem is then to compute the Pauli coefficients $(\alpha(P))_{P\in\mathcal{P}_n}$ efficiently. 

An exact expression for the Pauli coefficients is given by
 \begin{equation}\label{eq:alphatrace}
\alpha(P)=\frac{1}{2^n}\mathrm{tr}(A^{\dag}P) \quad (P\in \mathcal{P}_n),
 \end{equation}
which follows from \cref{eq:Aalpha} and the fact that $\mathcal{P}_n$ forms an orthonormal basis under the Hilbert-Schmidt inner product 
$\frac{1}{2^n}\mathrm{tr}(A^{\dag}B)$ on the vector space $\mathbb{C}^{2^n\times 2^n}$. Unfortunately, naive implementation of \cref{eq:alphatrace} suffers from poor scaling, $\mathcal{O}(2^{5n})$, and this equation is not easy to use for a derivation of an explicit linear-combination-of-unitaries decomposition of a Hamiltonian or a matrix used in specific applications.

In this work, we derive an exact and explicit representation of $\alpha$ based on a Walsh-Hadamard transform, defined as multiplication by the $n$th tensor power $H^{\otimes n}$ of the Hadamard matrix
\begin{equation}
    H \coloneqq  \begin{pmatrix}1 & 1 \\ 1 & -1\end{pmatrix}.
    \label{eq:Hadamard}
\end{equation}
Namely, for a matrix $A=(a_{p,q})_{p,q=0}^{2^n-1}\in \mathbb{C}^{2^n\times 2^n}$, the following theorem holds. 

\begin{theorem}[Pauli decomposition of $A\in\mathbb{C}^{2^n\times 2^n}$ via the Walsh-Hadamard transform]
\label{thm:main}
The following decomposition 
\begin{equation}\label{eq:result1}
A=\sum_{r,s=0}^{2^n-1}
\alpha_{r,s} P_{r,s},
\end{equation}
where
\begin{equation}\label{eq:result1a}
P_{r,s}\coloneqq \bigotimes_{j=0}^{n-1} \left(i^{r_j\wedge s_j} X^{r_j} Z^{s_j}\right)\in \mathcal{P}_n \quad (r,s\in [2^n-1]_0) 
\end{equation}
and 
\begin{equation}\label{eq:result2}
 \alpha_{r,s}\coloneqq  \frac{i^{-\lvert r\wedge s\rvert}}{2^n}\sum_{q=0}^{2^n-1}a_{q\oplus r,q}(H^{\otimes n})_{q,s} 
 \quad (r,s\in [2^n-1]_0),
\end{equation}
holds.  Conversely, $A$ can be recovered from its Pauli string decomposition according to 
\begin{equation}
\label{eq:result3}
    a_{r,s} = i^{\lvert \overline{r}\wedge s\rvert} \sum_{q=0}^{2^n-1} \alpha_{r \oplus s, q} (H^{\otimes n})_{q, s}.
\end{equation}
\end{theorem}

In Theorem~\ref{thm:main}, $\wedge$, $\oplus$, and $\bar{\cdot}$ are the bitwise AND, XOR, and NOT operators, respectively, $\lvert \cdot \rvert$ is the Hamming weight, $\{r_j\}_{j=0}^{n-1}$ and $\{s_j\}_{j=0}^{n-1}$ are the binary expansion coefficients of $r$ and $s$, respectively, and $[2^n-1]_0\coloneqq \{0,1,\ldots 2^n-1\}$. Precise definitions of these standard notations can be found at the beginning of \cref{sec:theorem_and_proof}. As we show in the proof of \cref{thm:main}, $P_{r,s}$ is indeed a Pauli string because $i^{r_j\wedge s_j}X^{r_j}Z^{s_j}\in \mathcal{P}\coloneqq \{I,X,Y,Z\}$ for $r_j, s_j\in \{0,1\}$. 

A naive implementation of matrix multiplication performed in \cref{eq:result2} scales as $\mathcal{O}(2^{3n})$. However, one can use the  Fast Walsh-Hadamard transform~\cite[Chapter 7]{yarlagadda_hadamard_1997}, which scales as $\mathcal{O}(n 2^{2n})$ and requires only $n 2^{2n}$ addition and $n 2^{2n}$ subtraction operations on complex numbers, with $\mathcal{O}(1)$ additional memory. Furthermore, because XOR transformation can be implemented in place, full Pauli decomposition can be implemented with $\mathcal{O}(1)$ memory overhead.

\textit{Relation to previous work.} To the best of our knowledge, the explicit \cref{eq:result2,eq:result3} and their proofs have not been published before. There are several algorithms developed for calculating the Pauli coefficients $\alpha$~\cite{romero_paulicomposer_2023,jones_decomposing_2024,koska_tree-approach_2024,ying_preparing_2023,lapworth2022hybridquantumclassicalcfdmethodology,gunlycke2020,hantzko_tensorized_2024,gidney2024stackoverflow,hamaguchi2024handbookefficientlyquantifyingrobustness}. 
To date, the best scaling methods achieve the time complexity $\mathcal{O}(n 2^{2n})$~\cite{gunlycke2020,hantzko_tensorized_2024,gidney2024stackoverflow,hamaguchi2024handbookefficientlyquantifyingrobustness} of the present algorithm. 
The previous fastest approach, as measured by CPU time, is the
optimized Rust implementation of the tensorised Pauli decomposition (TPD) algorithm of Hantzko, Binkowski, and Gupta~\cite{hantzko_tensorized_2024} in Qiskit~\cite{qiskit_tpd_2024,qiskit2024}. As shown in \cref{sec:num_results}, our approach outperforms this implementation. The most similar algorithm is that of Gidney~\cite{gidney2024stackoverflow} which we learnt about after posting this work online. It was proposed in Stack-Overflow comment~\cite{gidney2024stackoverflow} without though an explanation of how this algorithm was derived. Gidney's code then was also incorporated in Pennylane~\cite{bergholm_pennylane_2022}. The main difference of\cite{gidney2024stackoverflow} to ours is that our implementation carries out all operations, including XOR transformation, in-place. Hamaguchi, Hamada and Yoshioka proposed an algorithm~\cite{hamaguchi2024handbookefficientlyquantifyingrobustness} which is based on the Fast Walsh-Hadamard transformation. While the central primitive/subroutine in their algorithm is also the Fast Walsh-Hadamard transformation, we observe that our implementation is around 3x faster as shown in the results section. We also note that the (Fast) Walsh-Hadamard transformation is well understood and its fast implementations are known~\cite{fino_unified_1976,yarlagadda_hadamard_1997}. This transformation is used in other areas of quantum computation. For example, in the context of Pauli channels, it was shown that Pauli error rates and Pauli eigenvalues are related to each other through such a transformation~\cite{flammia_efficient_2020} and its fast implementation is used in several papers~\cite{harper_efficient_2020,harper_fast_2021}. However, as evident by \cref{eq:result2}, the Walsh-Hadamard transformation by itself does not provide the Pauli decomposition of a matrix as the application of XOR transformation and the correct phase factors must still be applied.

The paper is organised as follows. Our theoretical results, including a rigorous proof of Theorem~\ref{thm:main},  are contained in Section~\ref{sec:theorem_and_proof}. In this section, we also consider special cases where $A$ possesses certain symmetries, establishing three corollaries of Theorem~\ref{thm:main}. The numerical algorithm inspired by \cref{eq:result1,eq:result1a,eq:result2} is presented in Section~\ref{sec:num_results} together with comparison against the Qiskit implementation of TPD and the algorithm of Hamaguchi, Hamada and Yoshioka~\cite{hamaguchi2024handbookefficientlyquantifyingrobustness}. In Section~\ref{sec:conclusion}, we draw conclusions from the results presented herein and discuss the outlook for Pauli decomposition in the context of our results. 

\section{Theoretical Results}\label{sec:theorem_and_proof}

In this section, we prove our main result, Theorem~\ref{thm:main}. Additionally, we state and prove three corollaries concerning the special cases where the matrix to be decomposed is (i) Hermitian, (ii) real symmetric or (iii) complex symmetric.

To make our statements and proofs precise, we introduce some technology. Given a positive integer $N$, we define $[N]_0\coloneqq \{0,1,\ldots,N\}$. For any $p\in [2^n-1]_0$, there exists a unique binary decomposition
\begin{equation}\label{eq:pdecomposition}
p=\sum_{j=0}^{n-1} p_j 2^j \quad (p_j\in \{0,1\}),	
\end{equation}
which defines a map $p\mapsto (p_j)_{j=0}^{n-1}\in \{0,1\}^n$. Moreover, we may use \cref{eq:pdecomposition} to extend operations $\{0,1\}\times \{0,1\}\to \{0,1\}$ on bits to operations $[2^n-1]_0\times [2^n-1]_0\to [2^n-1]_0$ on $n$-bit integers, by bitwise application. Let $x,y\in \{0,1\}$ and recall the following operations 

\begin{equation}\label{eq:bitops}
\text{NOT:}\ \overline{x}\coloneqq  1-x,\quad \text{XOR:}\ x\oplus y\coloneqq \overline{x}y+x\overline{y},\quad \text{AND:}\ x\wedge y\coloneqq xy.
\end{equation}
Bitwise operations on integers are defined by composing the $\oplus$ (XOR) and $\wedge$ (AND) operators with \cref{eq:pdecomposition} yielding the definitions
\begin{equation}
\text{XOR:}\ p\oplus q \coloneqq \sum_{j=0}^{n-1} (p_j\oplus q_j)2^j, \quad \text{AND:}\ p\wedge q	\coloneqq \sum_{j=0}^{n-1} (p_j\wedge q_j)2^j \quad (p,q\in [2^n-1]_0).
\end{equation}
We also make use of the Hamming weight,
\begin{equation}
\lvert p\rvert \coloneqq \sum_{j=0}^{n-1} p_j \quad (p\in [2^n-1]_0). 
\end{equation}

\subsection{Proof of Theorem~\ref{thm:main}}

Let $E_{p,q}$ be a matrix unit for $\mathbb{C}^{2^n\times 2^n}$, i.e., 
\begin{equation}
E_{p,q}\coloneqq (\delta_{j,p}\delta_{k,q})_{j,k=0}^{2^n-1} \quad (p,q\in [2^n-1]_0).
\end{equation}
We write
\begin{equation}\label{eq:pqpjqj}
E_{p,q} = \bigotimes_{j=0}^{n-1} e_j(p,q),
\end{equation}
where
\begin{equation}\label{eq:ejpq}
e_{j}(p,q)\coloneqq (\delta_{k,p_j}\delta_{l,q_j})_{k,l=0}^1 \quad (j\in [n-1]_0;p,q\in [2^n-1]_0). 
\end{equation}

The Pauli matrices form a basis for $\mathbb{C}^{2\times 2}$; for the matrix units \cref{eq:ejpq}, the explicit decomposition is
\begin{equation}
e_j(p,q)=\frac12\big(X^{{p_j\oplus q_j}}+(-1)^{p_j}(i Y)^{p_j\oplus q_j}Z^{\overline{p_j\oplus q_j}}\big).    
\end{equation}
Next, using the identity $XZ=-iY$, we obtain
\begin{equation}\label{eq:pjqjXZ}
e_j(p,q)=
\frac12 \big(X^{p_j\oplus q_j}+(-1)^{q_j}X^{p_j\oplus q_j}Z\big).
\end{equation}
Putting \cref{eq:pjqjXZ} into \cref{eq:pqpjqj} yields
\begin{align}\label{eq:pq_as_IXZ}
 E_{p,q} =&\; \frac{1}{2^n} \bigotimes_{j=0}^{n-1} \big( X^{p_j\oplus q_j}+(-1)^{q_j}X^{p_j\oplus q_j}Z\big)= \frac{1}{2^n} \bigotimes_{j=0}^{n-1} \sum_{s_j=0,1} (-1)^{q_j\wedge s_j} X^{p_j\oplus q_j} Z^{s_j}= \frac{1}{2^n} \prod_{j=0}^{n-1} \sum_{s_j=0,1} (-1)^{q_j\wedge s_j} X_j^{p_j\oplus q_j} Z_j^{s_j},
\end{align}
where
\begin{equation}\label{eq:Paulij}
X_j\coloneqq I^{\otimes j} \otimes X  \otimes I^{\otimes (n-j-1)}, \quad Z_j\coloneqq I^{\otimes j} \otimes Z  \otimes I^{\otimes (n-j-1)} \quad (j\in [n-1]_0). 
\end{equation}
It follows from \cref{eq:pq_as_IXZ} that 
\begin{equation}
E_{p,q}=  \frac{1}{2^n}\Bigg( \sum_{s_0=0,1}\cdots \sum_{s_{n-1}=0,1}\Bigg)\prod_{j=0}^{n-1}   (-1)^{q_j\wedge s_j} X_j^{p_j\oplus q_j} Z_j^{s_j},
\end{equation}
which is equivalent to
\begin{equation}\label{eq:pqrsum}
E_{p,q}= \frac{1}{2^n} \sum_{s=0}^{2^n-1}\prod_{j=0}^{n-1}    (-1)^{q_j\wedge s_j} X_j^{p_j\oplus q_j} Z_j^{s_j} 
\end{equation}
with $s$ defined via \cref{eq:pdecomposition}.

To proceed, we will express the constant factors in \cref{eq:pqrsum} using tensor products of Hadamard matrices, \cref{eq:Hadamard}. Note that 
\begin{equation}
H_{q_j,s_j}=(-1)^{q_j\wedge s_j}   \quad (q_j,s_j\in \{0,1\}) 
\end{equation}
and hence,
\begin{equation}
(H^{\otimes n})_{q,s}=\prod_{j=0}^{n-1} (-1)^{q_j\wedge s_j} \quad (q,s\in [2^n-1]_0). 
\end{equation}
Using this observation in \cref{eq:pqrsum} gives
\begin{equation}\label{eq:pqrsum2}
E_{p,q}= \frac{1}{2^n} \sum_{s=0}^{2^n-1} \Bigg(\prod_{j=0}^{n-1} (-1)^{q_j\wedge s_j}\Bigg)
   \prod_{j=0}^{n-1}   X_j^{p_j\oplus q_j} Z_j^{s_j}=\frac{1}{2^n} \sum_{s=0}^{2^n-1} (H^{\otimes n})_{q,s}
   \prod_{j=0}^{n-1}   X_j^{p_j\oplus q_j} Z_j^{s_j}.
\end{equation}
Decomposing $A$ into a sum of elementary matrices and using \cref{eq:pqrsum2}, we write
\begin{equation}\label{eq:A1}
A=\frac{1}{2^n} \sum_{p,q,s=0}^{2^n-1} a_{p,q}(H^{\otimes n})_{q,s} \prod_{j=0}^{n-1} X_j^{p_j\oplus q_j}Z_j^{s_j}.
\end{equation}
Let us introduce a new variable of summation, $r\coloneqq p\oplus q$. Note that for each $q\in [2^n-1]_0$, the map $p\mapsto p\oplus q$ is a bijection $[2^n-1]_0\to [2^n-1]_0$ and that $p=q\oplus r$ and $r_j=p_j\oplus q_j$. Hence,
\begin{equation}\label{eq:A2}
A= \frac{1}{2^n}\sum_{r,s=0}^{2^n-1}\sum_{q=0}^{2^n-1}a_{q\oplus r,q}(H^{\otimes n})_{q,s} \prod_{j=0}^{n-1} X_j^{r_j} Z_j^{s_j}=\sum_{r,s=0}^{2^n-1}\Bigg(\frac{i^{-\lvert r\wedge s\rvert}}{2^n}\sum_{q=0}^{2^n-1}a_{q\oplus r,q}(H^{\otimes n})_{q,s}\Bigg) \Bigg(i^{\lvert r\wedge s\rvert} \prod_{j=0}^{n-1} X_j^{r_j} Z_j^{s_j}\Bigg),
\end{equation}
which concludes the proof of \cref{eq:result1,eq:result1a,eq:result2} if we can show that $P_{r,s}$, defined in \eqref{eq:result1a}, is an element of $\mathcal{P}_n$. Indeed, using the identity $XZ=-iY$, we compute
\begin{equation}
X^{r_j}Z^{s_j}=\begin{cases}
I & (r_j,s_j)=(0,0) \\
X & (r_j,s_j)=(1,0) \\
Z & (r_j,s_j)=(0,1) \\
-i Y & (r_j,s_j)=(1,1),
\end{cases}
\end{equation}
which shows that
\begin{equation}
i^{\lvert r\wedge s\rvert} \prod_{j=0}^{n-1} X_j^{r_j} Z_j^{s_j}=\bigotimes_{j=0}^{n-1} i^{r_j\wedge s_j}X^{r_j}Z^{s_j}\in \mathcal{P}_n. 
\end{equation}

To establish the converse, \cref{eq:result3}, we first write 
\begin{equation}\label{eq:result3calc}
a_{s\oplus r,s}=\frac{1}{2^{n}}\sum_{u=0}^{2^n-1}\sum_{q=0}^{2^n-1} a_{u\oplus r,u}(H^{\otimes n})_{u,q}(H^{\otimes n})_{q,s}=i^{\lvert r\wedge s\rvert}\sum_{q=0}^{2^n-1} \alpha_{r,q}(H^{\otimes n})_{q,s},
\end{equation}
using \cref{eq:result2} and the fact that $H^{\otimes n}$ is symmetric and satisfies $(H^{\otimes n})^2=2^n I^{\otimes n}$. By making the replacement $r\to r\oplus s$ in \eqref{eq:result3calc} and using the identity $(r\oplus s)\wedge s=\overline{r}\wedge s$, we obtain the result \cref{eq:result3}.

\subsection{Special Cases}
We consider the special cases where $A$ is (i) Hermitian, (ii) real symmetric, or (iii) (complex) symmetric. 

\begin{corollary}[Decomposition of Hermitian matrices]
\label{cor:Hermitian}
Suppose that $A\in \mathbb{C}^{2^n\times 2^n}$ satisfies $A^{\dag}=A$. Then, $\alpha_{r,s}\in \mathbb{R}$ for $r,s\in [2^n-1]_0$.
\end{corollary}

\begin{proof}
As Pauli strings are Hermitian, it follows from \cref{eq:result1} and our Hermiticity assumption that
\begin{equation}\label{eq:AAdag}
A-A^{\dag}= \sum_{r,s=0}^{2^n-1}({\alpha}_{r,s}-{\alpha}_{r,s}^*) \bigotimes_{j=0}^{n-1} \left(i^{r_j\wedge s_j} X^{r_j} Z^{s_j}\right)=0.
\end{equation}
By the linear independence of $\mathcal{P}_n$, the prefactors of the Pauli strings in \eqref{eq:AAdag} must be identically zero. Thus, ${\alpha}_{r,s}-{\alpha}_{r,s}^*=0$ for $r,s\in [2^n-1]_0$ and the result follows. 
\end{proof}

In the case of symmetric matrices, we can use the following proposition to make statements about the sparsity of $(\alpha_{r,s})_{r,s=0}^{2^n-1}$.

\begin{proposition}[Distribution of $\pm 1$ in $H^{\otimes n}$]
\label{prop:distribution}
The matrix $H^{\otimes n}\in \mathbb{C}^{2^n\times 2^n}$ has $2^{n-1}(2^n\pm 1)$ entries with values $\pm 1$, respectively.   
\end{proposition}

\begin{proof}
The claim is easily verified in the case $n=1$. Assume, inductively, that the claim holds for an arbitrary positive integer $n$. We write
$H^{\otimes (n+1)}=H^{\otimes n}\otimes H$;
a straightforward calculation shows that $H^{\otimes (n+1)}$ has $2^{n}(2^{n+1}\pm 1)$ entries with values $\pm 1$, respectively. The result follows. 
\end{proof} 

\begin{corollary}[Decomposition of real symmetric matrices]
\label{cor:realsymmetric}
Suppose that $A\in \mathbb{R}^{2^n\times 2^n}$ satisfies $A^{\top}=A$. Then, $\alpha_{r,s}\in \mathbb{R}$ when $
|r\wedge s|$ is even and $\alpha_{r,s}=0$ when $|r\wedge s|$ is odd. Correspondingly, $(\alpha_{r,s})_{r,s=0}^{2^n-1}$ has a sparsity of $(2^n-1)/2^{n+1}$. 
\end{corollary}

\begin{proof}
The first claim follows immediately from Corollary~\ref{cor:Hermitian}. To prove the second claim, we use the identity
\begin{equation}
\left(i^{r_j\wedge s_j} X^{r_j} Z^{s_j}\right)^{\top}
= (-1)^{r_j\wedge s_j}\left(i^{r_j\wedge s_j} X^{r_j} Z^{s_j}\right)
\end{equation}
to group Pauli strings in \cref{eq:result1}:
\begin{equation}
A-A^{\top}=\sum_{r,s=0}^{2^n-1} \big(\alpha_{r,s}-(-1)^{\lvert r\wedge s\rvert}\alpha_{r,s}\big)\bigotimes_{j=0}^{n-1}\left(i^{r_j\wedge s_j} X^{r_j} Z^{s_j}\right)=0.
\end{equation}

By the linear independence of $\mathcal{P}_n$, we must have $\alpha_{r,s}-(-1)^{r\wedge s}\alpha_{r,s}=0$ for $r,s\in [2^n-1]_0$. In the case $|r\wedge s|$ odd, this implies $\alpha_{r,s}=0$. 

To compute the sparsity of $(\alpha_{r,s})_{r,s=0}^{2^n-1}$, we note that 
\begin{equation}
(H^{\otimes n})_{r,s}=(-1)^{\lvert r\wedge s\rvert} \quad (r,s\in [2^n-1]_0),    
\end{equation}
i.e., $(H^{\otimes})_{r,s}$ is equal to the parity of $\lvert r\wedge s\rvert$. Hence, using Proposition~\ref{prop:distribution}, we obtain the stated sparsity: The locations of zeros correspond to the locations of $-1$ in $H^{\otimes n}$. 
\end{proof}

\begin{corollary}[Decomposition of complex symmetric matrices]
Suppose that $A\in \mathbb{C}^{2^n\times 2^n}$ satisfies $A^{\top}=A$. Then, $\alpha_{r,s}=0$ when $|r\wedge s|$ is odd. Correspondingly, $(\alpha_{r,s})_{r,s=0}^{2^n-1}$ has a sparsity of $(2^n-1)/2^{n+1}$.      
\end{corollary}

\begin{proof}
We write
\begin{equation}
A=\mathrm{Re}(A)+i\,\mathrm{Im}(A)    
\end{equation}
and note that $\mathrm{Re}(A)$ and $\mathrm{Im}(A)$ are real symmetric matrices. By Corollary~\ref{cor:realsymmetric} and linearity, we obtain the first claim. 

The proof of the sparsity claim is similar to that in the proof of Corollary~\ref{cor:realsymmetric}.
\end{proof}

\section{Algorithm and Numerical Results}\label{sec:num_results}

In this section, we present our implementation of the Pauli decomposition and compare its efficiency to the latest (still unreleased) Qiskit implementation~\cite{qiskit_tpd_2024} of the TPD algorithm~\cite{hantzko_tensorized_2024}. The following algorithm is also illustrated in Fig.~\ref{fig:algo}.

\begin{figure}
    \begin{overpic}[width=0.8\textwidth]{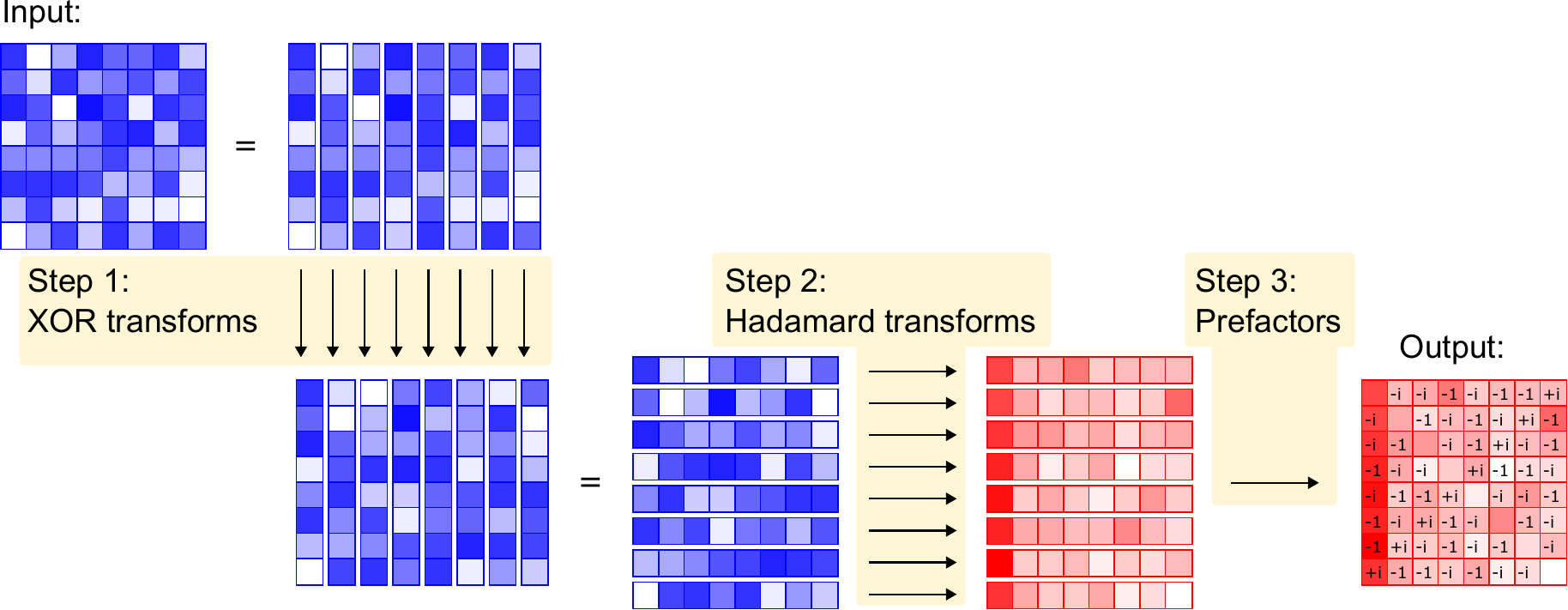}
    \put(6.2,37.6){$A=(a_{p,q})_{p,q=0}^{7}$}
    \put(96.5,16.7){$(\alpha_{r,s})_{r,s=0}^{7}$}
    \end{overpic}
    \caption{Illustration of Algorithm~\ref{algo:main} and Theorem~\ref{thm:main} for an $8\times8$ matrix. In Step 1, the XOR transform is applied, which permutes the elements in each matrix column according to~\eqref{eq:step1}. Then, in Step 2, the Walsh-Hadamard transform is applied to each row according to~\eqref{eq:step2}. This can be done efficiently with the Fast Walsh-Hadamard algorithm \cite{yarlagadda_hadamard_1997}. Finally, in Step 3, the elements are multiplied by prefactors according to~\eqref{eq:step3}.}
    \label{fig:algo}
\end{figure}

\begin{algo}[Pauli decomposition of $A\in\mathbb{C}^{2^n\times 2^n}$ via the Walsh-Hadamard transform]\label{algo:main} 
\hfill\break
\textbf{Input:} 
\begin{itemize}
    \item The matrix elements 
    $(a_{p,q})_{p,q=0}^{N-1}$ of a matrix $A$ of dimension $N\times N$ ($N=2^n$).
\end{itemize}
\textbf{Output:} 
\begin{itemize}
\item The coefficients of the Pauli decomposition $A=\sum_{r,s=0}^{2^n-1}
\alpha_{r,s} P_{r,s}$ in \cref{eq:result1}. They are computed in-place, i.e. they are stored in $(a_{p,q})_{p,q=0}^{N-1}$.
\end{itemize}
\textbf{Complexity:}
\begin{itemize}
    \item Runtime: $\mathcal{O}(N^2\log N)$
    \item Additional space: $\mathcal{O}(1)$
\end{itemize}
\textbf{Algorithm:}
\begin{enumerate}
\item XOR-Transform: Permute the elements of the matrix $A=(a_{p,q})_{p,q=0}^{N-1}$ according to
    \begin{equation}
        a_{r,q} \leftarrow a_{r\oplus q, q} \quad (r,q\in [N-1]_0).
    \label{eq:step1}
    \end{equation}
    This transformation can be done in place requiring only $2^{n-1} (2^{n} - 1)$ SWAP operations.
\item Fast Walsh-Hadamard Transform: Apply the Walsh-Hadamard transform
    \begin{equation}
        a_{r,s} \leftarrow \sum_{q=0}^{N - 1} a_{r,q} H^{\otimes n}_{q,s} \quad \quad (r,s\in [N-1]_0).
    \end{equation}
    This can be done in place using only $n {2^{2n}}$ addition and $n{2^{2n}}$ subtraction operations on complex numbers~\cite{yarlagadda_hadamard_1997}. This is because the generalised Hadamard matrix
    \begin{equation}
        H^{\otimes n} = \prod_{j=0}^{n - 1} H_j\ \text{with}\ 
    H_j\coloneqq I^{\otimes j} \otimes H  \otimes I^{\otimes (n-j-1)} \quad (j\in [n-1]_0),
    \label{eq:step2}
    \end{equation}
    and $H_j$ is highly sparse and structured: Each column of $H_j$ has only 2 non-zero elements, and they take values in $\{\pm1\}$. Other decompositions of generalised Hadamard matrices are possible~\cite{fino_unified_1976, yarlagadda_hadamard_1997}.
\item Prefactors: Multiply each element of $A$ according to 
    \begin{equation}
         a_{r,s} \leftarrow a_{r,s} \frac{(-i)^{|r \wedge s|}}{2^{n}} \quad \quad (r,s\in [N-1]_0).
    \label{eq:step3}
    \end{equation}
    This can be done in place.
\end{enumerate}
\textbf{Reference implementation:}
\begin{itemize}
    \item A reference implementation in C with Python \& Numpy bindings is freely available on Github \cite{efficient_pauli_decomp} under the MIT license. It can also be installed with \verb?pip install pauli_lcu?.
\end{itemize}
\end{algo}

The execution time of this algorithm for random Hermitian matrices as a function of $n$ is presented in~\cref{fig:execution_time} (green line). We compare our implementation to Qiskit's function \texttt{decompose\_dense} from \texttt{qiskit.\_accelerate.sparse\_pauli\_op} in the yet unpublished version \cite{qiskit_tpd_2024} (blue line). Qiskit includes calculation of the symplectic representation~\cite{aaronson_improved_2004} of the Pauli strings along with the coefficients; to make a fair comparison, we also show our algorithm including this calculation (orange line), and time our Python bindings rather than C code throughout. We observe that for more than 4 qubits, our implementation outperforms Qiskit and the maximum speed up we achieved is around a factor of $\approx 1.4$. We also compare our results with the approach of Hamaguchi, Hamada and Yoshioka, which is based on the Fast Walsh-Hadamard transformation as well. We used their C++ implementation as available at the GitHub repository~\cite{hamaguchi_code}. We observe that our algorithm provides a maximum speedup of a factor 3.6. For further details of our implementation we refer to our open-source implementation at Ref.~\cite{efficient_pauli_decomp}.
Finally, we note that the reverse transform \cref{eq:result3} can be achieved by effectively running Algorithm~\ref{algo:main} backwards.

\begin{figure}[ht]
    \centering
    \includegraphics[width=0.75\linewidth]{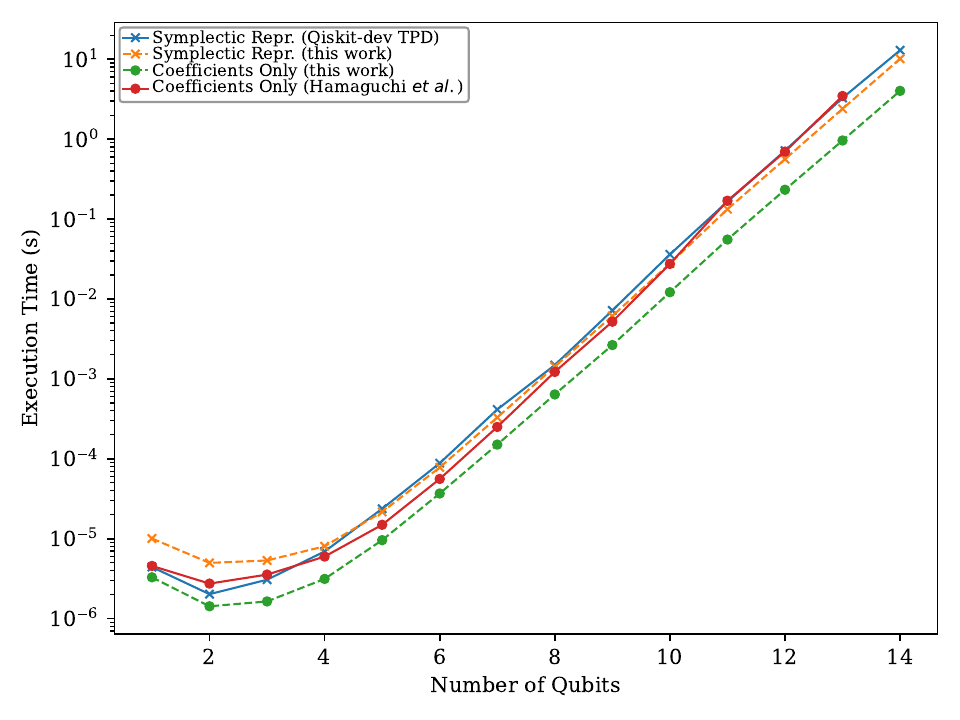}
    \caption{
    The execution time to perform a Pauli decomposition on random Hermitian matrices, as a function of the number of qubits. This calculation was performed on a single core AMD EPYC 7763 processor, 2.4GHz. Each data point was obtained by averaging over 10 runs for each of 10 different random matrices (100 runs contribute to each data point). The orange and green lines show this work, with and without calculation of a symplectic representation of Pauli strings, respectively. Qiskit-dev TPD refers to to the latest (still unreleased) Qiskit implementation~\cite{qiskit_tpd_2024}  of the TPD algorithm~\cite{hantzko_tensorized_2024}. Hamaguchi \textit{et al.} refers to the work~\cite{hamaguchi2024handbookefficientlyquantifyingrobustness, hamaguchi_code}
    }
    \label{fig:execution_time}
\end{figure}

\begin{table}[ht]
    \centering
\begin{tabular}{|c|c|c|c|c|}
\hline
Number of & Coefficients only & Symplectic repr.& Symplectic repr. &Coefficients only \\
basis functions & (this work) & (this work) & (Qiskit-dev TPD) &(Hamaguchi $et$ $al.$) \\
\hline
16 & 0.000877 & 0.00132 & 0.00188 & 0.00204 \\
\hline
32 & 0.0103 & 0.0269 & 0.0271 & 0.0327 \\
\hline
64 & 0.226 & 0.526 & 0.607 & 0.788 \\
\hline
128 & 3.95 & 9.52 & 11.1 & memory error \\
\hline
\end{tabular}
    \caption{The execution time (in seconds) to perform a Pauli decomposition on electronic integrals in a molecular orbital basis set. This calculation was performed on a single core AMD EPYC 7763 processor, 2.4GHz.
    Qiskit-dev TPD refers to the Qiskit implementation~\cite{qiskit_tpd_2024}  of the TPD algorithm~\cite{hantzko_tensorized_2024}. Hamaguchi \textit{et al.} refers to the work~\cite{hamaguchi2024handbookefficientlyquantifyingrobustness, hamaguchi_code}.
    }
    \label{tab:molecular}
\end{table}

To evaluate the performance of our algorithms further, we consider two more examples from chemistry and physics. One is the transformation of the electron repulsion integrals:
\begin{equation}
    h_{IJ} = \int\limits_{\mathbb{R}^3\times \mathbb{R}^3} d{\bf r'}d{\bf r} \, \frac{\rho_I^*({\bf r},{\bf r'}) \rho_J({\bf r},{\bf r'})}{|{\bf r} - {\bf r'}|},\quad I,J \in [D^2-1]_0,
\end{equation}
where $D$ is the number of basis functions, and $\rho_I({\bf r},{\bf r'})$ is a pair-density matrix. We use precomputed integrals for the Fe$_2$S$_2$ complex from Ref~\cite{georges_fcidump_files}. \cref{tab:molecular} shows the execution time of our algorithms, the algorithm of Hantzko \textit{et al.}~\cite{hantzko_tensorized_2024} as implemented on Qiskit, and the algorithm of Hamaguchi \textit{et al.}~\cite{hamaguchi2024handbookefficientlyquantifyingrobustness,hamaguchi_code}. For the largest basis set size, we observe a speedup of 1.16 for the symplectic representation , and further factor of 2.4 (2.78 in total as compared to Qiskit-dev TPD) when one is interested in Pauli coefficients only.

Our second example is the transformation of the kinetic energy operator represented in real space using the dual plane-wave basis~\cite{skylaris2002nonorthogonal,mostofi2002total, babbushLowDepthQuantumSimulation2018a}. For a cubic cell with lattice constant of one and for uniformly distributed grid points, the kinetic energy matrix is
\begin{equation}
    T_{{\bf n},{\bf n'}} = 
    2\pi^2 
    \sum_{m_1,m_2,m_3=-N^{1/3}/2}^{N^{1/3}/2-1} 
    \left(m_1^2 + m_2^2 +m_3^2\right)
    \exp\left(
    \frac{2\pi i}{N^{1/3}} 
    \sum_{j=1}^{3} m_j (n_j-n'_j)
    \right)
\end{equation}
where $N$, the total number of grid points, is a power of $2^3$; $n_1, n_2, n_3, n'_1, n'_2, n'_3 \in \left[N^{1/3} - 1\right]_0$, and ${\bf n} = (n_1, n_2, n_3)$, ${\bf n'} = (n'_1, n'_2, n'_3)$. In addition to being symmetric, the kinetic energy matrix also contains translational symmetry. We provide the python script for generating such a matrix in Ref.~\cite{georges_2024_13844352}. \cref{tab:kinetic} shows the execution time of different algorithms. For the largest calculations, where the total number of grid points is $N = 2^{15}$, we observe a factor of 1.65 speedup for the symplectic representation as compared to Qiskit-dev TPD, and a further factor of 2.5 when one is interested in Pauli coefficients only (a factor of 4.12 speed up in total as compared to Qiskit-dev TPD).

\begin{table}[ht]
    \centering
\begin{tabular}{|c|c|c|c|c|}
\hline
Number of & Coefficients only & Symplectic repr.& Symplectic repr. &Coefficients only \\
grid points & (this work) & (this work) & (Qiskit-dev TPD) &(Hamaguchi $et$ $al.$) \\
        
\hline
8 & $2.95 \cdot 10^{-5}$ & $2.13 \cdot 10^{-5}$ & 0.000102 & $5.62 \cdot 10^{-5}$ \\
\hline
64 & $6.23 \cdot 10^{-5}$ & $8.35 \cdot 10^{-5}$ & 0.00018 & $8.5 \cdot 10^{-5}$ \\
\hline
512 & 0.00227 & 0.00546 & 0.0113 & 0.00706 \\
\hline
4096 & 0.221 & 0.528 & 0.941 & 0.79 \\
\hline
32768 & 17.2 & 41.3 & 81.6 & memory error \\
\hline
\end{tabular}
    \caption{The execution time (in seconds) to perform a Pauli decomposition on a kinetic energy matrix in real space. This calculation was performed on a single core AMD EPYC 7763 processor, 2.4GHz.
    Qiskit-dev TPD refers to the Qiskit implementation~\cite{qiskit_tpd_2024}  of the TPD algorithm~\cite{hantzko_tensorized_2024}. Hamaguchi \textit{et al.} refers to the work~\cite{hamaguchi2024handbookefficientlyquantifyingrobustness, hamaguchi_code}
    }
    \label{tab:kinetic}
\end{table}

As can be seen from~\cref{fig:algo}, our algorithm allows for straightforward parallelization. First, the XOR transformation (Step 1 in \cref{fig:algo}) is applied to each column, then Hadamard transformation and phase factor multiplication are applied to each row independently (Steps 2--3 in \cref{fig:algo}). We have implemented this parallelization with OpenMP and tested it on the largest example from~\cref{tab:kinetic} ($2^{15} \times 2^{15}$ matrix). As is shown on~\cref{fig:openmp_speedup}, we observe nearly ideal speedup for 2 and 4 cores. For 8 cores, the wall time is reduced by a factor of 7 and the Pauli coefficients can be calculated in around 2.5 seconds as compared to 17.3 seconds on a single AMD EPYC 7763 core.

\begin{figure}[ht]
    \centering
    \includegraphics[width=0.5\linewidth]{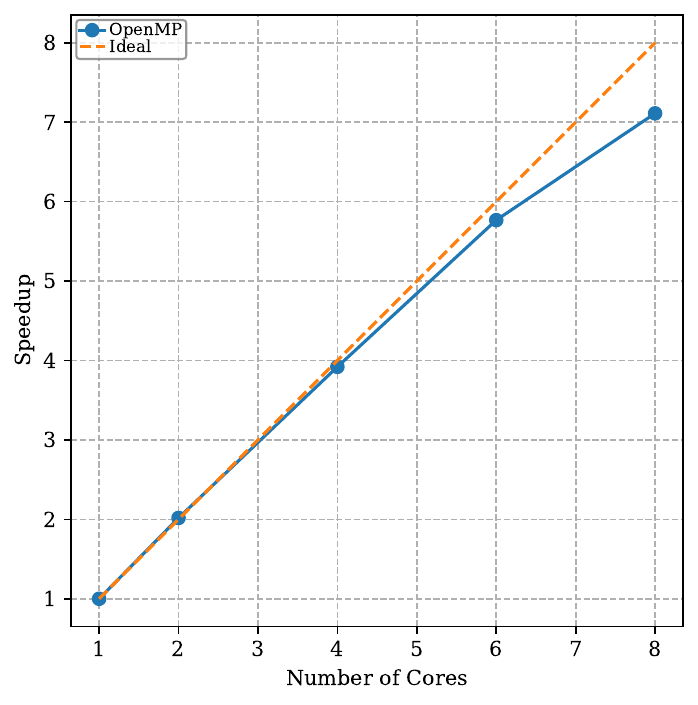}
    \caption{
    Speedup of Pauli decomposition due to parallelization using OpenMP. The same $2^{15} \times 2^{15}$ kinetic energy matrix as in~\cref{tab:kinetic} was used as an example. This calculation was performed on AMD EPYC 7763 processors, 2.4GHz, and 386 GB RAM. 
    }
    \label{fig:openmp_speedup}
\end{figure}

\section{Conclusion}\label{sec:conclusion}

In this work, we have presented an exact and explicit approach to the Pauli decomposition of matrices $A\in \mathbb{C}^{2^n\times 2^n}$. Practically, our theoretical result inspires an efficient algorithm for the computation of the Pauli string coefficients. We develop this algorithm and demonstrate that it outperforms the previous leading solutions in the literature.
We emphasize that, as we employ the fast Walsh-Hadamard transform, our algorithm achieves favorable performance versus existing solutions without sophisticated optimisation techniques.

Our results have immediate applications in quantum chemistry.
Pauli decompositions of Hamiltonians in second quantization are available through fermion-to-qubit mappings such as that of Jordan-Wigner~\cite{Jordan1928,ortiz_quantum_2001}. However, they are not applicable to Hamiltonians in first quantization.
Our theorem applied to Hermitian matrices and their higher-dimensional tensorial extensions allows one to derive the explicit Pauli decomposition of first quantization Hamiltonians for the study of chemical systems on quantum computers~\cite{georges_quantum_2024}. 
More generally, we expect that our results will be useful in the block-encoding of matrices. Finally, our solution is available as an open source package at~\cite{efficient_pauli_decomp}. 

\section{Supplementary Information}
Scripts used to produce results for this paper are available at Zenodo~\cite{georges_2024_13844352}. Our code is further available on Github~\cite{efficient_pauli_decomp} and with \verb?pip install pauli_lcu?.

\section*{Acknowledgments}
We thank Marius Bothe, R\'obert Izs\'ak, and Matt Ord for collaboration on closely related topics. AVI also thanks Mark Turner for helpful discussion about parallelisation. The work presented in this paper was partially funded by a grant from Innovate UK under the `Commercialising Quantum Technologies: CRD \& Tech round 2’ competition (Project Number 10004793). 

\bibliography{main.bib}

\end{document}